\documentclass[11pt]{article}
\pdfoutput=1  
\usepackage[utf8]{inputenc}
\usepackage{amsmath,amssymb,amsthm}
\usepackage{graphicx}
\usepackage{booktabs}
\usepackage{natbib}
\usepackage[margin=1in]{geometry}
\usepackage{hyperref}
\usepackage{xcolor}
\usepackage{placeins}

\newtheorem{proposition}{Proposition}

\newcommand{\elpd}{\mathrm{elpd}}
\newcommand{\khat}{\hat{k}}
\newcommand{\Var}{\mathrm{Var}}
\newcommand{\E}{\mathbb{E}}
\newcommand{\N}{\mathcal{N}}
\newcommand{\GFF}{G_{FF}}
\newcommand{\GBB}{G_{BB}}
\newcommand{\GBF}{G_{BF}}
\newcommand{\RB}{\textsc{rb}}

\title{\bf Partial pooling predicts cross-validation reliability:\\
a closed-form triage and Rao--Blackwellised cure for hierarchical LOO}
\author{Aidan D. Bindoff\\[2pt]
Wicking Dementia Research \& Education Centre,\\
University of Tasmania, Hobart, Tasmania, Australia\\[2pt]
\small \texttt{aidan.bindoff@utas.edu.au} \quad ORCID: \texttt{0000-0002-0943-2702}}
\date{\today}

\begin{document}
\maketitle

\begin{abstract}
For hierarchical models, Pareto-smoothed importance-sampling LOO (PSIS-LOO) fails on
exactly the folds where a random-effect coordinate is data-driven and its group is
small. We show that the Gelman--Pardoe pooling factor and the associated structural
leverage predict these folds from the model structure and group sizes, without
forming importance weights. In Gaussian linear mixed models the leverage reduces to
group size, giving a design-time map that separates the $\khat_i>0.7$ folds with AUC
$0.96$; across replicated logistic GLMMs the post-fit, weight-free predictor reaches
AUC $0.81$. The cure for a flagged fold is integrated importance sampling
(marginalise the random-effect block, importance-sample only the base parameters),
the standing recommendation of the \texttt{loo} documentation and available in closed
or quadrature form in the latent Gaussian literature \citep{vehtari2016glvm,
merkle2019bayesian, burkner2021efficient}. We contribute its observation-level
specialisation for random-intercept GLMMs, an analytic Gaussian downdate and a 1-D
quadrature for Bernoulli, binomial and Poisson responses, packaged as a drop-in
\texttt{rb\_loo(fit)} for the Stan/\texttt{brms}/\texttt{loo} ecosystem. Against exact
refits, this marginalised estimator (RB-LOO) is $3\times$ more accurate than moment
matching on singleton-heavy logistic GLMMs. On an overdispersed count dataset with
$97$ failing folds, moment matching leaves $37$ uncorrected and is no more accurate
than raw PSIS-LOO, while RB-LOO reproduces the $82$-minute exact refit (elpd RMSE
$0.04$) at no cost. The error changes decisions: comparing this model against a
negative-binomial alternative, PSIS-LOO reports decisive evidence ($z=4.9$) and
\texttt{reloo} at the recommended threshold reports significant evidence ($z=3.4$)
for the more complex model, where an exact fold-wise analysis, reproduced by RB-LOO,
finds the two indistinguishable ($z=1.0$). A base--fiber decomposition explains the
split: a Schur complement separates case-deletion influence into a vertical (pooling)
term, which governs where PSIS-LOO fails, and a horizontal (variance-component) term,
which governs where RB-LOO's own residual importance sampling is strained. The
resulting two-level triage (pooling factor, then residual base-$\khat$) recovers the
exact answer while refitting only the few folds that need it. The marginalisation is
not new; the contributions are the a-priori pooling-factor triage, the packaged
closed forms, the head-to-head showing marginalisation dominates the
importance-sampling repairs, and the geometric account.
\end{abstract}

\noindent\textbf{Keywords:} cross-validation; leave-one-out; PSIS; hierarchical
models; partial pooling; Rao--Blackwellisation; importance sampling; pooling factor.

\section{Introduction}

Leave-one-out cross-validation (LOO-CV) is a standard tool for Bayesian model
assessment, and Pareto-smoothed importance sampling
\citep[PSIS-LOO;][]{vehtari2017practical, vehtari2024pareto} has made it cheap
enough for routine use: a single model fit yields an estimate of the expected
log predictive density (elpd) for every held-out observation, together with a
per-fold diagnostic $\khat_i$ that flags when the importance-sampling
approximation is unreliable. The diagnostic is a safety valve: when
$\khat_i > 0.7$, the recommendation is to refit the model with observation
$i$ removed \citep{vehtari2017practical}.

For hierarchical models this safety valve fires often, and in a structured way.
When an observation belongs to a small group (a singleton, in the limit), deleting
it removes most of the information about that group's random effect, so the
leave-one-out posterior of that coordinate moves far from the full-data posterior,
the importance weights acquire a heavy tail, and $\khat_i$ crosses the threshold.
The recommended cure, an exact refit, is what PSIS-LOO was meant to avoid; the
cheaper alternative, moment matching \citep{paananen2021implicitly}, still requires
a per-fold optimisation.

This paper makes three points, validated on replicated simulations and real data
and packaged as a drop-in tool.

\paragraph{C1 (predict).} The pooling factor of \citet{gelman2006weakly}, the
fraction of a group's posterior precision contributed by the prior, and the
associated per-observation structural leverage predict which folds PSIS-LOO fails
on, without forming importance weights. In a Gaussian linear mixed model the
per-observation Fisher information is constant, so the leverage reduces to group
size and the predictor is a design-time map, available before the data are seen; it
separates the $\khat_i>0.7$ folds with AUC $0.96$. For GLMMs the predictor uses the
fitted variance component and Fisher weights, so it requires one model fit, but
still no importance weights and no case deletion, and reaches AUC $0.81$ on
replicated logistic models with heterogeneous group sizes. We are not aware of prior
work mapping the pooling factor to PSIS-LOO reliability.

\paragraph{C2 (cure), and its provenance.} The cure for a flagged fold is to
integrate out the random-effect block and importance-sample only the base
parameters. The idea is not ours: it is integrated importance-sampling LOO, the
marginal-likelihood LOO of \citet{merkle2019bayesian}, the latent-marginalised LOO
with quadrature of \citet{vehtari2016glvm}, the closed forms of
\citet{burkner2021efficient}, and the standing recommendation of the \texttt{loo}
documentation for hierarchical PSIS failures. Marginalising the deleted coordinate,
rather than reweighting it, removes the heavy-tailed weight. We contribute the
closed-form observation-level specialisation for random-intercept GLMMs (an analytic
Gaussian rank-1 downdate; a 1-D quadrature for Bernoulli, binomial and Poisson),
packaged as a drop-in \texttt{rb\_loo(fit)}, and the head-to-head that the C1 flag
makes possible: on the folds that matter the marginalised estimator matches exact
refits at zero optimisation cost, is $3\times$ more accurate than moment matching
\citep{paananen2021implicitly} on singleton-heavy logistic GLMMs, and corrects
failures that moment matching leaves in place on a real count dataset. We call the
estimator RB-LOO (Rao--Blackwellised LOO): integrated IS is Rao--Blackwellisation of
the LOO estimator (Section~\ref{sec:theory}), so the name describes a standard
operation.

\paragraph{C3 (triage and explain).} A base--fiber bundle structure
\citep{bindoff2026fibr} organises the two failure modes. A Schur complement
(Section~\ref{sec:theory}) splits each observation's case-deletion influence into a
vertical (pooling) term, which governs where PSIS-LOO fails and which C1 exploits,
and a horizontal (variance-component) term, which governs where RB-LOO's residual
base importance sampling is strained. This gives a two-level triage (pooling factor,
then residual base-$\khat$) that recovers the exact LOO answer while refitting only
the few percent of folds where RB-LOO is itself unreliable.

The random-effect marginalisation (C2) is prior art. The structural leverage is
classical case-influence machinery \citep{wei1998generalized} in its modern
hierarchical form \citep{lovison2026augmented}; that influential observations produce
high $\khat$ is known \citep{peruggia1997variance, vehtari2017practical}; and INLA's
cross-validatory CPO computes the same marginal LOO predictive deterministically
\citep{rue2009inla, held2010posterior}. Our contribution is the composition: the
a-priori pooling-factor-to-PSIS-reliability map, the packaged closed forms, the
two-level triage, the head-to-head the flag enables, and the geometric account.
Section~\ref{sec:prior} draws the boundaries.

\section{Background and related work}\label{sec:prior}

\paragraph{PSIS-LOO and the $\khat$ diagnostic.}
Write the full-data posterior as $p(\theta \mid y)$ and the pointwise likelihood
as $p(y_i \mid \theta)$. PSIS-LOO estimates
$\elpd_i = \log p(y_i \mid y_{-i})$ by importance-reweighting the full-data
posterior draws with weights $w_i^{(s)} \propto 1/p(y_i \mid \theta^{(s)})$,
Pareto-smoothing the largest weights, and reporting the fitted Pareto shape
$\khat_i$ as a reliability diagnostic \citep{vehtari2017practical,
vehtari2024pareto}. The importance-sampling variance is infinite when the
leave-one-out posterior has heavier tails than the proposal in the relevant
direction; $\khat_i > 0.7$ signals this. \citet{peruggia1997variance} first
characterised the case-deletion importance weight variance.

\paragraph{The pooling factor.}
For a group $j$ with random effect $\alpha_j$, prior precision $1/\sigma_u^2$ and
data precision $I_j$ (the Fisher information the group's observations carry about
$\alpha_j$), the Gelman--Pardoe pooling factor is
$\pi_j = (1/\sigma_u^2)/(I_j + 1/\sigma_u^2)$, the fraction of the posterior
precision of $\alpha_j$ contributed by the prior \citep{gelman2006weakly}. A
singleton in a weakly-informative hierarchy has $\pi_j$ near $1$ (fully pooled,
data-starved); a large group has $\pi_j$ near $0$. The per-observation
\emph{structural leverage} $h_i = I_i/(I_j + 1/\sigma_u^2)$, with within-group
Fisher additivity $I_j=\sum_{i\in j}I_i$, satisfies $\sum_{i \in j} h_i = 1 - \pi_j$
\citep{bindoff2026fibr}, and is the shrinkage-weighted observation leverage of
\citet{lovison2026augmented}.

\paragraph{Conditional versus marginal LOO, and integrated importance sampling.}
A hierarchical model admits two LOO estimands \citep{merkle2019bayesian}: the
\emph{conditional} predictive $p(y_i \mid \theta, \alpha_{j(i)})$ integrated over
the leave-one-out posterior, and the \emph{marginal} (integrated) predictive
$p(y_i \mid y_{-i})$ with the random effects marginalised. These coincide as
estimands (integrating the conditional likelihood over the joint leave-one-out
posterior, which includes $\alpha$, gives the marginal predictive), but they
differ sharply as importance-sampling targets: PSIS over the conditional likelihood
reweights the full posterior in the $\alpha$ direction, which is exactly where the
data-driven random effect makes the weights heavy-tailed, whereas marginalising
$\alpha$ first removes that direction from the importance sampling. This
\emph{integrated importance sampling} is well established. \citet{merkle2019bayesian}
argue for the marginal predictive and note its better-behaved weights;
\citet{vehtari2016glvm} compute LOO for latent Gaussian models by marginalising the
latent value analytically (Gaussian likelihood) or by 1-D quadrature (logistic,
Poisson), the same machinery this paper specialises to grouped GLMMs;
\citet{burkner2021efficient} give closed-form integrated LOO with a rank-1 downdate
for non-factorised normal and Student-$t$ models; and the \texttt{loo} documentation
recommends integrating out group-specific parameters, ``analytically or by
quadrature,'' precisely to avoid PSIS failure in hierarchical models. INLA computes
the same marginal LOO predictive deterministically as its cross-validatory CPO
\citep{rue2009inla, held2010posterior}, and leave-group-out CV
\citep{liu2022leave} targets a related group-level predictive. RB-LOO is a
closed-form, triage-equipped instance of this integrated-IS family, operating
post-hoc on existing MCMC draws within the \texttt{loo} ecosystem; throughout we
compare only to refits that target the same (marginal) estimand, so the comparison
is fair.

\paragraph{Cures for high $\khat$.}
Given a high-$\khat$ fold, the alternatives to a refit are all
importance-sampling repairs that keep the conditional target and fix the proposal.
Moment matching \citep{paananen2021implicitly} applies an affine transformation to
the posterior draws for each high-$\khat$ observation, chosen by a per-fold
optimisation; mixture / implicitly-adaptive importance sampling
\citep{silva2024robust} builds a mixture proposal with guaranteed finite variance
for the many-high-$\khat$ regime. We take moment matching (in its robust
split-transformation form, the recommended default) as the primary state-of-the-art
comparator, and contrast the mechanism with mixture IS: these repair the proposal
for the \emph{conditional} weights, whereas RB-LOO removes the offending direction
from the problem by marginalisation. Our empirical finding is that on
singleton-driven folds the affine repair cannot represent a random effect reverting
to its prior, so marginalisation dominates.

\paragraph{Hierarchical leverage and influence.}
The generalised leverage $g_i^\top G^{-1} g_i$ and its Schur decomposition are
classical \citep{wei1998generalized}; the hierarchical case is treated by
\citet{hodges2001counting}, \citet{cui2010partitioning} (who partition effective
degrees of freedom and cite the pooling factor) and, most directly,
\citet{lovison2026augmented}, whose augmented hat matrix gives the subject-level
structural leverage. These works hold the variance components fixed (Lovison) or
measure complexity rather than influence (Cui--Hodges), and none addresses
cross-validation reliability; Section~\ref{sec:theory} demarcates our claim
accordingly.

\section{Theory: the orthogonal decomposition of case-deletion influence}\label{sec:theory}

We work in the base--fiber coordinates of the hierarchical bundle
\citep{bindoff2026fibr}: base $\theta$ (hyperparameters: population effects,
variance components, GP hypers), fiber $\alpha$ (group coordinates), with Fisher
metric
\[
G = \begin{pmatrix} \GBB & \GBF \\ G_{FB} & \GFF \end{pmatrix},
\]
$\GFF$ block-diagonal across groups. An observation $i$ touches only its own
group, so its score splits as $g_i = (g_i^B, g_i^F)$ with $g_i^F$ supported on the
$\alpha_{j(i)}$ block.

\begin{proposition}[Orthogonal leverage identity, T1]\label{thm:t1}
With $A = -\GFF^{-1} G_{FB}$, $\tilde g_i = g_i^B + A^\top g_i^F$ the horizontal
lift of the score, and $M = \GBB - \GBF \GFF^{-1} G_{FB}$ the Schur complement (the
marginal base precision),
\[
g_i^\top G^{-1} g_i
= \underbrace{g_i^{F\top}\GFF^{-1} g_i^F}_{\text{vertical (pooling)}}
+ \underbrace{\tilde g_i^\top M^{-1}\tilde g_i}_{\text{horizontal (base leverage)}}.
\]
\end{proposition}
\begin{proof}
The block inverse of $G$ is
$G^{-1} = \left(\begin{smallmatrix} M^{-1} & -M^{-1}\GBF\GFF^{-1}\\ -\GFF^{-1}G_{FB}M^{-1} & \GFF^{-1}+\GFF^{-1}G_{FB}M^{-1}\GBF\GFF^{-1}\end{smallmatrix}\right)$.
Expanding $g_i^\top G^{-1} g_i$ with $g_i=(g_i^B,g_i^F)$ and collecting the four
blocks gives $g_i^{F\top}\GFF^{-1}g_i^F + (g_i^B-\GBF\GFF^{-1}g_i^F)^\top M^{-1}(g_i^B-\GBF\GFF^{-1}g_i^F)$;
since $A^\top = -\GBF\GFF^{-1}$, the second term is $\tilde g_i^\top M^{-1}\tilde g_i$.
\end{proof}
It is a purely algebraic (Schur) identity, requiring only that $G$ and $G_{FF}$ be
invertible (equivalently $M$ invertible); block-diagonality of $G_{FF}$ is not
needed here, only later for the per-group decomposition. The reproduction script
confirms it to $4\times10^{-16}$. Geometrically it is the Pythagorean split of the
score into a vertical (fiber) and a horizontal (base) component under the
metric-orthogonal connection $A$; we use that reading only as interpretation. The vertical term summed over a group equals
$1-\pi_j$ (T3, an exact consequence of the Gelman--Pardoe definition); it is the
\emph{capacity} for a group's coordinate to move under deletion. The horizontal term
is the influence propagated into the hyperparameter posterior.

The next proposition justifies the name Rao--Blackwellised. It is a
variance-reduction statement, stated against the correct baseline: the full-marginal
importance-sampling estimator of the same marginal LOO predictive, not the
conditional estimand.

\begin{proposition}[Rao--Blackwell variance reduction]\label{thm:rb}
Fix observation $i$. The PSIS-LOO importance weight is
$W_i(\theta,\alpha)=1/p(y_i\mid\theta,\alpha_{j(i)})$, and under the full-data
posterior $\E[W_i]=1/p(y_i\mid y_{-i})$, the reciprocal of the marginal LOO
predictive (the standard case-deletion identity). Let $\bar W_i(\theta)=\E[W_i\mid\theta]
= 1/p(y_i\mid\theta,y_{-i})$ be the base-conditional expectation, which is the
integrand RB-LOO forms by marginalising $\alpha$. Both $W_i$ and $\bar W_i$ are
unbiased for $\E[W_i]$, and, with all expectations and variances taken under the
full-data posterior,
\[
\Var\big(\bar W_i\big)\;\le\;\Var\big(W_i\big),
\]
by the Rao--Blackwell (conditional-Jensen) inequality, with equality iff $W_i$ is
$\theta$-measurable.
\end{proposition}
\begin{proof}
$\E[\bar W_i]=\E[\E[W_i\mid\theta]]=\E[W_i]$ by the tower rule, and
$\Var(W_i)=\Var(\E[W_i\mid\theta])+\E[\Var(W_i\mid\theta)]\ge\Var(\bar W_i)$, with
equality iff $\Var(W_i\mid\theta)=0$ almost surely.
\end{proof}
The inequality is useful when the excess variance lives in the fiber. If the
small-group random effect makes $\Var(W_i)=\infty$ while the base is well identified,
so that $\Var(\bar W_i)=\Var(\E[W_i\mid\theta])<\infty$, marginalising replaces an
infinite-variance importance-sampling problem with a finite-variance one. That
finiteness is a condition on base identifiability, not a corollary of the inequality,
which otherwise reads $\infty\le\infty$. The reciprocal-and-log step to elpd inherits
the usual self-normalised consistency, not unbiasedness, and PSIS smoothing sits on
top of either estimator. Empirically (Section~\ref{sec:experiments}) the marginalised
weights have lighter tails and low base-$\khat$ where the conditional weights fail.

\begin{proposition}[The residual-strain predictor, T2]\label{thm:t2}
To first order (one-step case deletion), the variance of the log RB-LOO base
importance weight for deleting observation $i$, under the base posterior with
covariance $M^{-1}$, is $\tilde g_i^\top M^{-1}\tilde g_i$, the horizontal (base)
leverage. Hence the base leverage is an approximate, closed-form predictor of the
residual base-$\khat$ after Rao--Blackwellisation.
\end{proposition}
Empirically the one-step base influence predicts the actual $\sigma_u$
leave-one-out shift at correlation $0.987$ and the residual base-$\khat$ at Spearman
$0.62$--$0.71$ in the strained regime; being first-order, it under-shoots the two
largest shifts, so we use it as a predictor and diagnostic, not an exact identity.
As a practical flag the residual base-$\khat$ itself is used directly, and
Section~\ref{sec:experiments} shows either quantity catches the folds where RB-LOO
diverges from an exact refit (AUC $0.90$--$0.93$).

\paragraph{Worked example (Gaussian random intercept).} With base
$(\mu,\tau=\log\sigma_u)$, fiber $\alpha$, $\sigma$ known, and residual
$r_i = y_i-\mu-\alpha_j$, the direct base score has zero $\tau$-component (an
observation carries no first-order information about $\sigma_u$ on its own), yet
its leave-one-out influence on $\sigma_u$ is real and flows through the connection:
\[
\tilde g_{i,\tau} \;\propto\; r_i\,\alpha_j / (\sigma_u^2\, G_{FF,j}),
\]
proportional to the residual times the group effect, amplified by $M^{-1}_{\tau\tau}\sim 1/J$
for few groups. This is the singleton/outlier-group pathology, derived rather than
assumed.

\begin{proposition}[RB-LOO regimes, T4]\label{thm:t4}
Conditional on the base draw, the RB-LOO integrand is: \emph{analytically exact}
for a Gaussian random-effect block (the rank-1 downdate of
Section~\ref{sec:method}); \emph{numerically exact} for a scalar Bernoulli,
binomial or Poisson random-intercept block, where the one-dimensional integral over
$\alpha_j$ is evaluated by deterministic quadrature to grid resolution (independent
of any auxiliary-variable representation); and \emph{approximate} for a non-linear
block (a smooth change-point, a GP length-scale), with error controlled by the
non-Gaussianity of the conditional and degrading predictably
(Section~\ref{sec:boundary}). The remaining Monte Carlo error, in every case, is
the base importance sampling, diagnosed by the base-$\khat$ and bounded by
Proposition~\ref{thm:t2}.
\end{proposition}

\paragraph{Demarcation.} The structural leverage of Theorem~\ref{thm:t1}'s
vertical term is \citet{lovison2026augmented}'s subject hat value; we do not claim
it. What is not in that literature is (i) the case-deletion influence on the
variance components (the horizontal term, which lies outside the
variance-components-fixed hat matrix) and (ii) the cross-validation
application: the pooling factor as a PSIS-reliability predictor and the base
leverage as an RB-residual predictor. We claim the CV application and the
composition, and cite \citet{lovison2026augmented}, \citet{cui2010partitioning}
and \citet{wei1998generalized} to bound it.

\section{Method: pooling-factor triage and RB-LOO}\label{sec:method}

\paragraph{Triage (C1).} Compute the pooling factor $\pi_j$ and the
per-observation structural leverage $h_i$ from the Fisher blocks. Two regimes
differ. In a Gaussian LMM the per-observation Fisher information is the constant
$1/\sigma^2$, so $h_i$ depends only on group size (and, through $\sigma_u$, a single
fitted scalar): the ranking of folds is a design quantity, available before the
outcomes are seen. In a GLMM the Fisher weight $I_i$ depends on the fitted mean, so
$h_i$ is evaluated at the posterior mean of one model fit; it is still formed
without importance weights and without case deletion, but it is post-fit, not
pre-data. The triage is thus design-time relative to the cross-validation, not
before fitting. Folds with high $h_i$ are flagged as likely PSIS-LOO failures. Where
group sizes are constant (as in both real-data models of Section~\ref{sec:ereal})
the group-size content of $h_i$ is degenerate and the flag's discrimination comes
from the fitted-mean content, a point we return to in Section~\ref{sec:ereal}.

\paragraph{RB-LOO (C2).} For each observation $i$ in group $j$ and each base draw
$\phi^{(s)} = (\beta^{(s)}, \sigma_u^{(s)}, \dots)$, form the RE-marginalised
leave-one-out predictive
\[
L^{(s)}_{\RB,i} = \log \int p(y_i \mid \alpha, \phi^{(s)})\,
p(\alpha \mid y_{j\setminus i}, \phi^{(s)})\, d\alpha,
\]
and feed the matrix $L_{\RB}$ to \texttt{loo} for the base-only PSIS estimate. For
a \textbf{Gaussian} block the integral is an analytic rank-1 downdate: with
$s^2=\sigma^2$, downdated group precision $P_{-i}=1/\sigma_u^2 + (n_j-1)/s^2$ and
downdated mean $\hat\alpha_{-i} = P_{-i}^{-1}\sum_{o\in j\setminus i}(y_o-\eta_o)/s^2$,
\[
L^{(s)}_{\RB,i} = \log \N\!\big(y_i \,\big|\, \eta_i + \hat\alpha_{-i},\; s^2 + P_{-i}^{-1}\big).
\]
For \textbf{Bernoulli, binomial and Poisson} blocks the 1-D integral over
$\alpha_j\sim\N(0,\sigma_u^2)$ is evaluated by a self-normalised quadrature on a
grid rescaled per draw by $\sigma_u^{(s)}$ (the shipped package uses $64$ nodes over
$\pm 6\sigma_u$ by default, widened for singleton observation-level effects with
extreme responses), at negligible cost and exact to grid resolution. Because
$\alpha_j$ is integrated out rather than importance-reweighted, the fiber-driven
heavy tail is gone and the residual importance sampling is over the well-behaved
base. The real-data results (Section~\ref{sec:ereal}) call the shipped
\texttt{rb\_loo(fit)}; the simulation studies use a reference implementation of the
same estimator, which agrees with it up to the quadrature grid.

\paragraph{Refit flag (C3).} Folds whose residual base-$\khat$ still exceeds $0.7$
after Rao--Blackwellisation are sent to an exact refit. For the single
random-intercept GLMMs of Sections~\ref{sec:emm}--\ref{sec:ereal} this set is
empty (one deletion barely moves a well-identified base), so RB-LOO alone
suffices. In the few-groups, poorly-identified-base regime the set is non-empty and
the flag earns its place: Section~\ref{sec:basestress} shows the residual
base-$\khat$ catches, with AUC $0.93$, exactly the folds where RB-LOO diverges from
an exact refit, so the two-level triage recovers the exact answer while refitting
only a few percent of folds.

\paragraph{Software.} The \texttt{rbloo} package exposes this as a single call
\texttt{rb\_loo(fit)}, an S3 method over \texttt{brmsfit} / \texttt{stanreg}
returning, per observation, the pooling factor and structural leverage (a-priori
triage), the RB-LOO pointwise elpd and base $\khat$, and the refit flag. It
targets the \texttt{loo} ecosystem as a drop-in.

\section{Experiments}\label{sec:experiments}

All simulations are replicated and reported with Monte Carlo standard error
(MCSE); point estimates from single datasets are noisy at the group sizes that
produce PSIS failures. Software versions: \texttt{R} 4.6.0, \texttt{loo} 2.9.0,
\texttt{brms} 2.23.0 with \texttt{cmdstan}/\texttt{rstan}; the Gaussian case uses
a hand-written conjugate Gibbs sampler.

\paragraph{What ``matches exact refit'' does and does not show.} On a singleton
fold RB-LOO and \texttt{reloo} compute the same marginal integral (RB-LOO by
quadrature over the random-effect prior, \texttt{reloo} by refitting and resampling
the dropped level from that same prior), so agreement there confirms that RB-LOO
computes the intended quantity accurately but cannot by itself falsify it. The
informative comparison on those folds is against PSIS-LOO and moment matching, which
target the same quantity and get it wrong. The independent test, where \texttt{reloo}
can falsify RB-LOO because RB-LOO's own base importance sampling is an approximation,
is the few-groups experiment of Section~\ref{sec:basestress}, where RB-LOO diverges
from the exact refit on the folds its diagnostic flags. The singleton experiments
establish that RB-LOO beats the importance-sampling alternatives; the few-groups
experiment establishes when RB-LOO itself must defer to a refit.

\subsection{Simulation studies}\label{sec:sim}
Two replicated studies establish C1 and C2, one conjugate and one not. The Gaussian
random-intercept LMM ($60$ replicates of singleton-heavy designs, $\sigma_u=1.3$;
$2880$ folds, $128$ PSIS failures; conjugate Gibbs, with a brute-force refit gold
standard on a leverage-stratified subset) is the design-time case, where the
per-observation Fisher information is constant and structural leverage reduces to
group size. The logistic GLMM ($25$ replicates, $50$ groups of heterogeneous size;
$2036$ folds, $191$ failures; MCMC fits, \texttt{reloo} gold) is the post-fit case,
where the leverage's group-size content still contributes because group sizes vary.
Structural leverage separates the $\khat_i>0.7$ folds a priori (AUC $0.96$ and
$0.81$), and RB-LOO cures every failure at zero refit cost, with elpd RMSE against
exact refit well below PSIS-LOO on the high-$\khat$ folds (Table~\ref{tab:sim};
Figure~\ref{fig:g1} shows the Gaussian case). We report per-replicate Spearman with
MCSE as primary, since folds within a fit are not independent and the pooled AUC is
optimistic about effective sample size.

\begin{table}[htbp]\centering\small
\caption{Simulation studies. AUC is pooled over folds; Spearman is per-replicate
(mean $\pm$ MCSE). RMSE is elpd against exact refit on the high-$\khat$
folds.}\label{tab:sim}
\begin{tabular}{lccccc}
\toprule
study & failures & C1 Spearman & C1 AUC & cured & RMSE (RB / PSIS)\\
\midrule
Gaussian LMM ($60$ reps)  & $128$ & $0.75\pm0.01$ & $0.96$ & $100\%$ & $0.096 / 0.232$\\
Logistic GLMM ($25$ reps) & $191$ & $0.61\pm0.03$ & $0.81$ & $100\%$ & $0.015 / 0.060$\\
\bottomrule
\end{tabular}
\end{table}

\begin{figure}[tbp]\centering
\includegraphics[width=\textwidth]{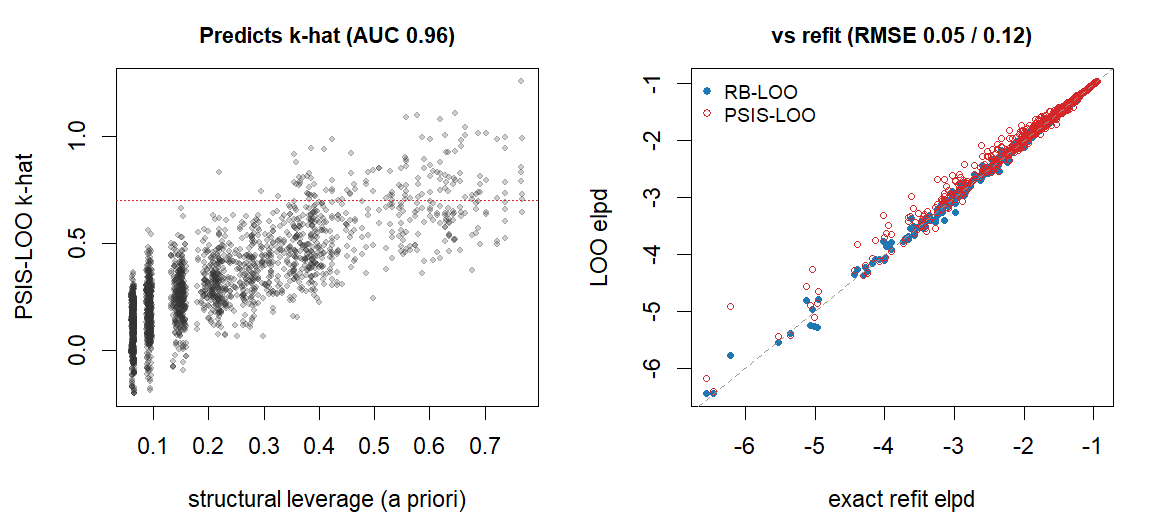}
\caption{Gaussian LMM. Left: a-priori structural leverage predicts PSIS-LOO $\khat$
(AUC $0.96$). Right: RB-LOO and PSIS-LOO elpd against the brute-force refit, RMSE
$0.050$ and $0.123$.}\label{fig:g1}
\end{figure}

\subsection{Moment-matching head-to-head (E-MM)}\label{sec:emm}
We refit the logistic GLMM design in \texttt{brms} (six replicate datasets, $77$
high-$\khat$ folds) and compared PSIS-LOO, moment matching
\citep{paananen2021implicitly}, and RB-LOO on the high-$\khat$ folds against
\texttt{reloo} gold. Moment matching performs one optimisation per high-$\khat$ fold;
RB-LOO performs none. Per-replicate elpd RMSE against the exact refit (mean $\pm$
MCSE) is $0.094\pm0.025$ for PSIS-LOO, $0.081\pm0.026$ for moment matching, and
$0.032\pm0.008$ for RB-LOO, so RB-LOO is about $2.5\times$ more accurate than moment
matching (pooled over folds, $3\times$: $0.140$, $0.131$, $0.044$). Both RB-LOO and
moment matching drive every $\khat_i$ below $0.7$; the difference is accuracy and
cost. Moment matching barely improves on PSIS-LOO ($0.131$ vs $0.140$) because an
affine transformation of the conditional draws cannot represent a data-driven
singleton random effect reverting to its prior under deletion, which is what RB-LOO
marginalises in closed form. Figure~\ref{fig:emm} shows the high-$\khat$ folds
against the exact refit.

\begin{figure}[tbp]\centering
\includegraphics[width=0.55\textwidth]{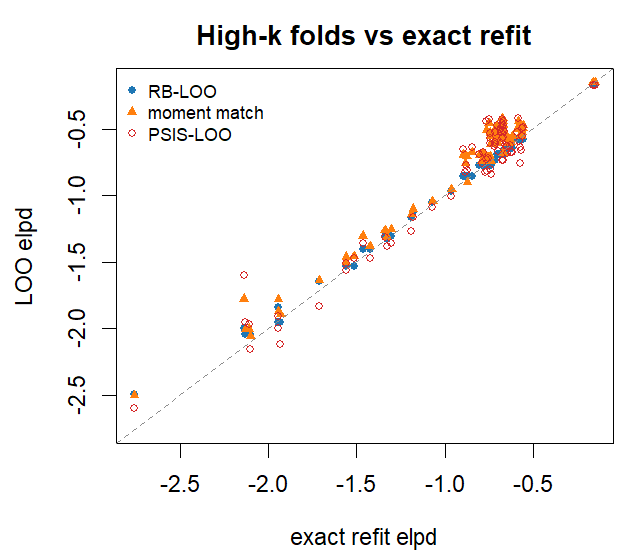}
\caption{Moment-matching head-to-head (E-MM), logistic GLMM. On the high-$\khat$
folds, RB-LOO (blue) sits on the exact-refit diagonal while moment matching (orange)
and PSIS-LOO (red) scatter above it.}\label{fig:emm}
\end{figure}

\subsection{Real data: epilepsy (E-real)}\label{sec:ereal}
The \texttt{brms::epilepsy} data holds $236$ seizure counts on $59$ patients over
four visits, overdispersed and a known PSIS-LOO stressor. A patient-effect Poisson
model, $\texttt{count}\sim\texttt{zBase}*\texttt{Trt}+\texttt{zAge}+(1\mid\texttt{patient})$
evaluated with one \texttt{rb\_loo(fit)} call, produces only five PSIS failures
(four visits per patient). RB-LOO cures them and is the most accurate estimator
against \texttt{reloo} (elpd RMSE $0.72$, versus PSIS-LOO $1.76$ and moment matching
$0.77$), but with constant group size the leverage's discrimination is fitted-mean
driven rather than group-size driven (AUC $0.84$); the design-time claim is carried
by the simulations of Section~\ref{sec:sim}. The observation-level model below is
the real stressor.

\paragraph{The overdispersion stressor.} Overdispersion in count data is standardly
modelled by an observation-level random effect (a Poisson-lognormal),
$\texttt{count}\sim\texttt{zBase}*\texttt{Trt}+\texttt{zAge}+(1\mid\texttt{obs})$, in
which every observation is its own singleton group: the maximally data-driven
regime. Here PSIS-LOO fails on \textbf{97 of 236} folds (max $\khat=1.34$).
\begin{itemize}
\item \textbf{C1.} Structural leverage predicts the failures with Spearman $0.60$
and AUC $0.80$ (with this many failures $\khat$ is no longer floored, so the rank
correlation is strong too). The groups are again all singletons, so the predictor is
monotone in the fitted count (large counts are influential) rather than
group-size driven.
\item \textbf{C2.} RB-LOO cures all $97$ (mean $\khat$ $0.68\to-0.01$; max
$1.34\to0.35$). Moment matching leaves $37$ of the $97$ folds with $\khat>0.7$; it
does not fix the singleton-driven failures.
\item \textbf{Accuracy} against \texttt{reloo} on all $97$ high-$\khat$ folds
(elpd RMSE): PSIS-LOO $0.638$, moment matching $0.658$, RB-LOO $\mathbf{0.041}$,
so RB-LOO is about $15\times$ more accurate than either and moment matching is here
no better than raw PSIS-LOO. RB-LOO matches the exact refit to within $0.04$ nats;
PSIS-LOO and moment matching are systematically over-optimistic, most on the
extreme-count folds (Figure~\ref{fig:ereal_olre}, right).
\item \textbf{Cost.} RB-LOO performs $0$ optimisations; moment matching performs
$97$ ($8$ min); \texttt{reloo} performs $97$ exact refits ($82$ minutes) for the
answer RB-LOO returns at once.
\end{itemize}
On a recognised PSIS-LOO stressor, then, the state-of-the-art importance-sampling
cure misses a third of the failures and is no more accurate than the incumbent,
while RB-LOO reproduces the exact refit at no cost (Figure~\ref{fig:ereal_olre}).

\begin{figure}[tbp]\centering
\includegraphics[width=\textwidth]{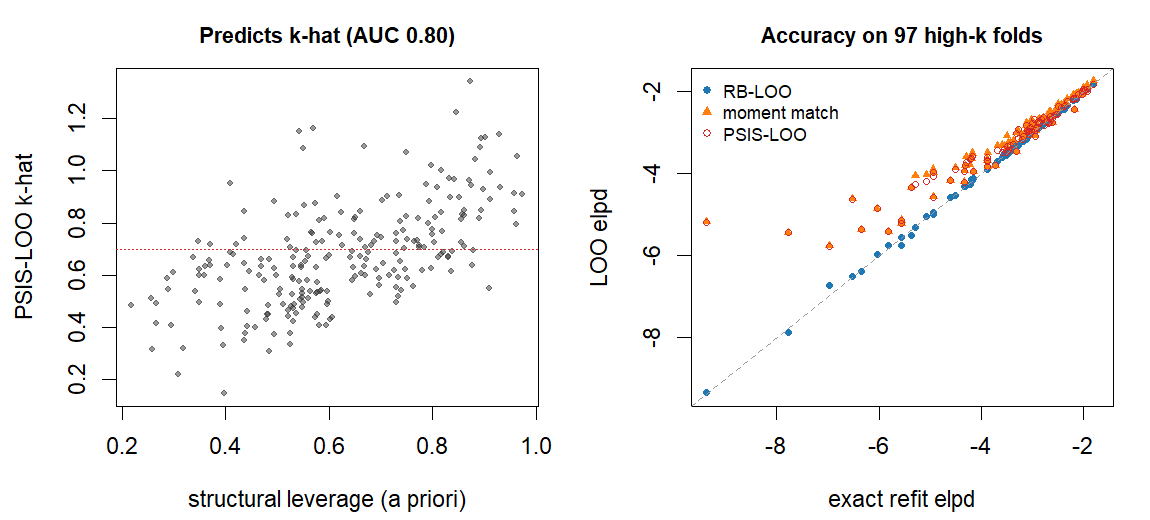}
\caption{Epilepsy, observation-level overdispersion model. Left: structural leverage
predicts PSIS-LOO $\khat$ (AUC $0.80$) across $97$ failures. Right: against the exact
refit, RB-LOO (blue) matches while PSIS-LOO (red) and moment matching (orange) are
over-optimistic, most on the extreme-count folds (RMSE $0.041$ against $0.638$ and
$0.658$).}\label{fig:ereal_olre}
\end{figure}

\subsection{The over-optimism changes the model-selection verdict (E-decision)}\label{sec:decision}
Does the elpd error matter for a decision? We compare two standard models for the
same overdispersed counts: $M_1$, the observation-level-effect Poisson-lognormal
above, and $M_2$, a negative-binomial regression
($\texttt{count}\sim\texttt{zBase}*\texttt{Trt}+\texttt{zAge}$, \texttt{negbinomial}),
which captures overdispersion without a per-observation latent. These are
near-equivalent predictive models, but $M_1$'s per-observation random effect makes
PSIS-LOO fail on $97$ of its folds while $M_2$'s PSIS-LOO is clean (no $\khat>0.7$).
Their elpd difference is
\begin{itemize}
\item PSIS-LOO (naive default): $\Delta\elpd = +39.3\pm8.0$ ($z=4.9$), $M_1$
decisively preferred;
\item \texttt{reloo} at the recommended $\khat>0.7$ threshold:
$\Delta\elpd = +11.0\pm3.2$ ($z=3.4$), $M_1$ still significantly preferred;
\item RB-LOO (zero refits): $\Delta\elpd = +3.1\pm3.1$ ($z=1.0$), the two models
indistinguishable.
\end{itemize}
The \texttt{loo\_compare} standard error $\sqrt{N}\,\mathrm{sd}(d_i)$ is itself
uncertain and mildly anti-conservative \citep{sivula2020uncertainty}, so we read the
$z$-values as descriptive; the robust finding is the shift $39.3\to11.0\to3.1$, which
does not depend on the standard error. To adjudicate, we exact-refit $40$ folds
stratified across the $\khat$ range. RB-LOO matches the exact refit fold-by-fold
(RMSE $0.017$, mean bias $+0.001$), whereas PSIS-LOO is over-optimistic with a bias
that grows in $\khat$ (mean $+0.088$; $+0.207$ on the $\khat>0.7$ folds), and is
already biased on the $0.55$--$0.7$ folds ($+0.092$) that \texttt{reloo} at threshold
$0.7$ does not refit. The standard \texttt{reloo} remedy therefore inherits residual
PSIS over-optimism on its sub-threshold folds; its $+11.0$ is itself optimistic, and
RB-LOO's $+3.1$, quadrature-converged (stable to $0.1$ nat across grids from
$\pm6\sigma$ to $\pm16\sigma$) and validated fold-wise against exact refits, is the
correct verdict. The comparison is consistent: each model's elpd is computed by a
method accurate for that model. $M_1$'s failing per-observation effect is what RB-LOO
corrects (validated to the exact refit); $M_2$ has no random effect, its PSIS-LOO is
reliable (max $\khat=0.51$), and exact-refitting all five of its $\khat>0.2$ folds
moves its total by $0.0$ nats, so $M_2$'s elpd is itself exact. RB-LOO is applied
only where the error is. The asymmetry is intrinsic: PSIS-LOO's over-optimism biases
model comparison when the candidates differ in latent structure (here, a
latent-variable model against a marginal one), where the per-model bias does not
cancel in the difference. A practitioner using default PSIS-LOO, or the recommended
\texttt{reloo} threshold, would report decisive or significant evidence for the more
complex model where an exact analysis finds none (Figure~\ref{fig:decision}).

\begin{figure}[tbp]\centering
\includegraphics[width=\textwidth]{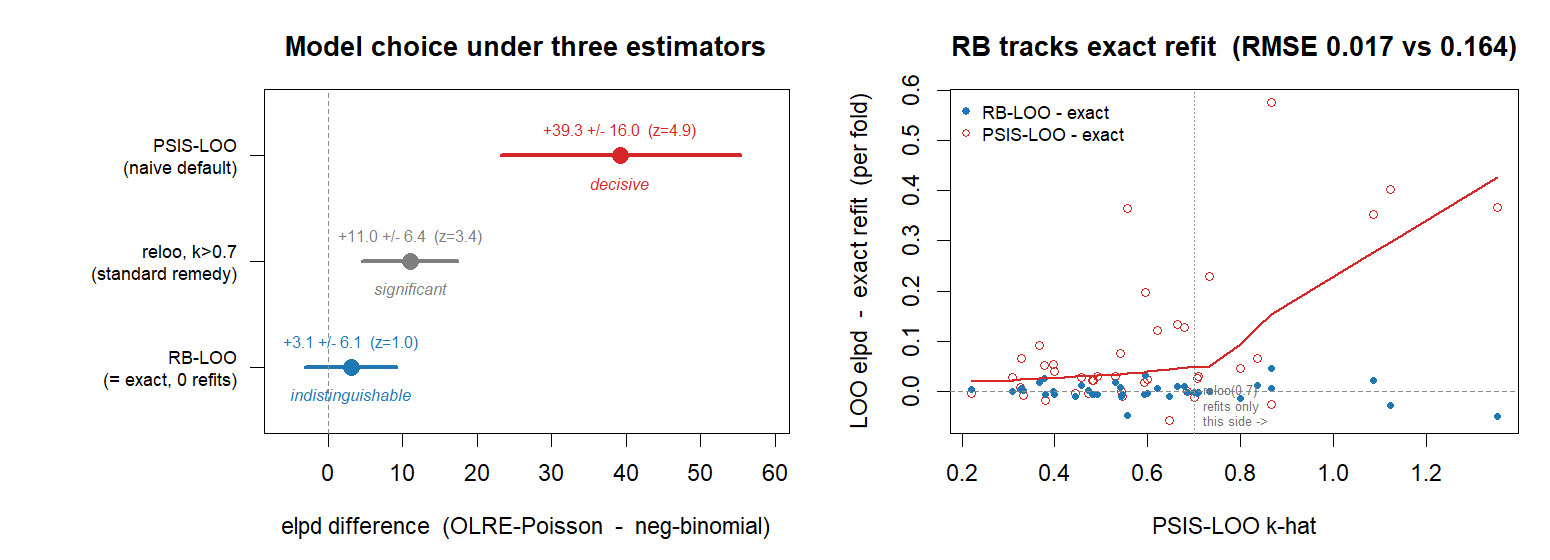}
\caption{Model selection under three estimators (E-decision), observation-level
Poisson against negative-binomial on epilepsy. Left: PSIS-LOO calls the difference
decisive ($z=4.9$), \texttt{reloo} at threshold $0.7$ significant ($z=3.4$), RB-LOO
indistinguishable ($z=1.0$). Right: against $40$ exact refits, RB-LOO (blue) is
unbiased across the $\khat$ range while PSIS-LOO (red) is over-optimistic, biased
even below $0.7$ where \texttt{reloo} does not refit.}\label{fig:decision}
\end{figure}

\subsection{When RB-LOO itself fails, and the refit flag (E-basestress)}\label{sec:basestress}
The experiments above stress PSIS-LOO; this one stresses RB-LOO. We use few-groups
Gaussian LMMs ($J=4$ to $6$ groups with singletons, weakly-identified $\sigma_u$;
$2340$ folds over replicated designs), where deleting one observation moves the base
posterior appreciably, so RB-LOO's residual base importance sampling is a real
approximation and a brute-force refit of every fold is an independent oracle that can
falsify it.
\begin{itemize}
\item RB-LOO's residual base-$\khat$ exceeds $0.7$ on $65$ folds, so RB-LOO is
strained here, unlike the singleton experiments.
\item RB-LOO is not exact on the flagged folds: elpd RMSE against the exact refit is
$0.068$ on the unflagged folds (base-$\khat\le0.7$) but $0.532$ on the flagged folds
(base-$\khat>0.7$), so RB-LOO's own diagnostic tracks its own error.
\item The residual base-$\khat$ separates the folds where RB-LOO diverges from the
exact refit (by more than $0.25$ nats) with AUC $0.93$; the base leverage (T2) does
so at AUC $0.90$.
\item Applying RB-LOO on the unflagged folds and an exact refit on the flagged folds
gives elpd RMSE $0.067$ against exact, refitting only $3\%$ of folds, against
$0.111$ for RB-LOO alone and $0.389$ for PSIS-LOO.
\end{itemize}
Here \texttt{reloo} does reveal RB-LOO error, where the flag predicts, and the
two-level triage delivers refit-quality LOO at a few percent of the refit cost
(Figure~\ref{fig:basestress}).

\begin{figure}[tbp]\centering
\includegraphics[width=\textwidth]{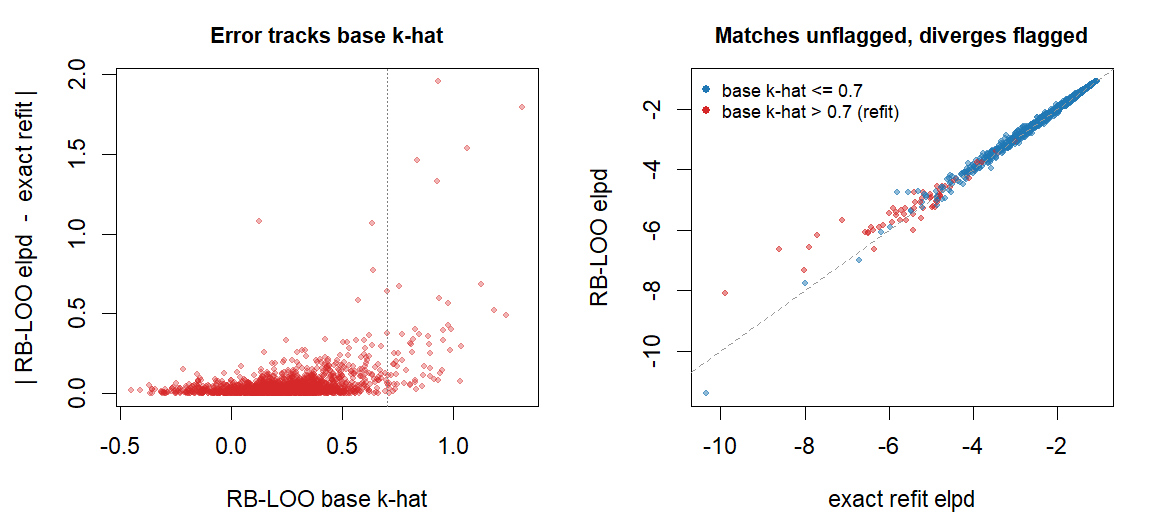}
\caption{RB-LOO's own limit and the refit flag (E-basestress), few-groups Gaussian
LMMs. Left: RB-LOO's elpd error against the exact refit grows with its own residual
base-$\khat$. Right: RB-LOO matches the exact refit on the unflagged folds (blue)
and diverges on the flagged folds (red), which the flag catches.}\label{fig:basestress}
\end{figure}

\subsection{The linearisation boundary (E-M5)}\label{sec:boundary}
Beyond the base importance sampling, the other approximation is the RB integrand
itself, which is analytic for Gaussian blocks and numerically exact for the scalar
GLMM families but only approximate for non-linear blocks. In a smooth logistic
change-point model (the \texttt{bjlm} setting) we mapped where the linearised
conditional degrades. A moment-matched Gaussian conditional stays calibrated across
sparsity and sharpness (Kolmogorov--Smirnov statistic $0.018$ to $0.043$), and the
raw Fisher/pooling approximation is adequate except at extreme change-point sharpness
($\rho=20$: KS $0.060$ against the $0.048$ critical value), where the mode$\neq$mean
gap and the dropped curvature bite. The degradation is monotone and predictable,
which is the limit of the method for non-conjugate blocks.

\subsection{Efficiency in the number of draws (E-Sscaling)}\label{sec:sscaling}
The ``zero refit cost'' claim concerns model refits; a second efficiency axis is the
number of posterior draws $S$. On a failing fold PSIS-LOO's importance weights have
(near-)infinite variance, so its elpd estimate is noisy and biased and improves only
slowly in $S$; RB-LOO marginalises the offending direction, so its weights have finite
variance and its estimate converges quickly. On the same Gaussian LMM design, drawing
one long conjugate chain and subsampling its first $S$ draws (so only $S$ varies), we
measure elpd RMSE against exact refit on the high-$\khat$ folds over
$S\in\{250,\dots,8000\}$ (Figure~\ref{fig:sscaling}). RB-LOO at $S=250$ (RMSE $0.17$)
is already more accurate than PSIS-LOO at $S=8000$ (RMSE $0.27$), roughly an order of
magnitude fewer draws, and over this range PSIS-LOO improves far more slowly and
plateaus well above RB-LOO. PSIS-LOO remains consistent as $S\to\infty$, but its
convergence on infinite-variance folds is slow enough that in practice it does not
reach RB-LOO's accuracy.

\begin{figure}[tbp]\centering
\includegraphics[width=0.62\textwidth]{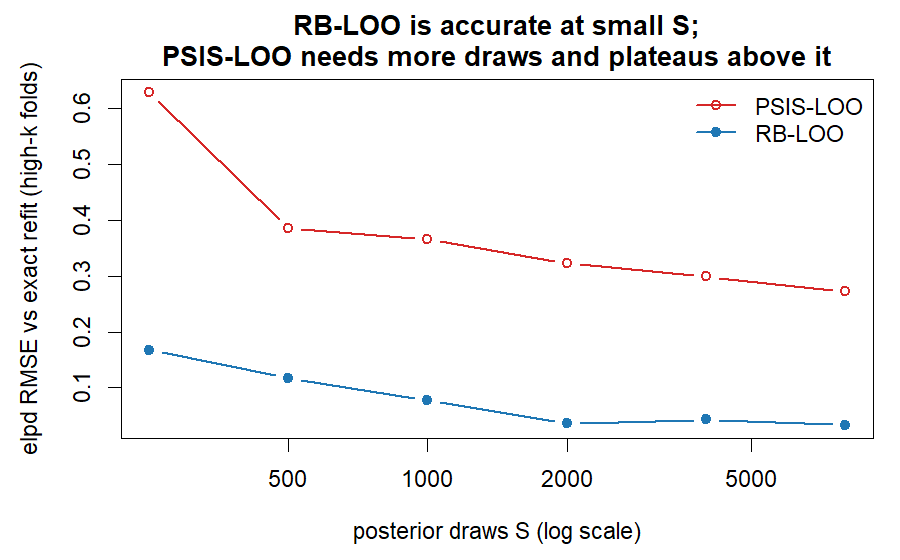}
\caption{Efficiency in the number of draws (E-Sscaling). On the high-$\khat$ folds,
RB-LOO reaches refit-quality accuracy with far fewer posterior draws than PSIS-LOO,
which plateaus above it.}\label{fig:sscaling}
\end{figure}

\FloatBarrier

\section{Discussion}

RB-LOO applies to a single random-intercept grouping factor with a scalar
random-effect block: Gaussian LMMs (analytic) and Bernoulli, binomial and Poisson
GLMMs (1-D quadrature). Its estimand is the integrated (marginal-over-RE)
predictive, against which we compare throughout \citep{merkle2019bayesian}. Random
slopes and correlated within-group effects are a low-dimensional extension, where
the quadrature becomes two- or three-dimensional but stays cheap. Crossed or
multiple grouping factors are the hard case: a deleted observation loads on several
fibers at once and there is no single block to collapse, so the package falls back
to plain PSIS-LOO, which we name as out of scope rather than degrade silently.
Practically, the method is a two-line change to a hierarchical LOO workflow: triage
on the pooling factor, integrate out the flagged folds, and refit only the residual
few, which removes most of the refits PSIS-LOO would otherwise demand. The geometry
is interpretation: the pooling factor is the vertical half of a case-deletion
influence identity whose horizontal half is the variance-component funnel, and
cross-validation reliability reads off the two halves.

\section{Honesty and limitations}
\begin{itemize}
\item \textbf{The marginalisation cure is not novel.} It is integrated
importance-sampling LOO \citep{merkle2019bayesian, vehtari2016glvm,
burkner2021efficient}, the standing \texttt{loo} recommendation for hierarchical
PSIS failures, and the same estimand as INLA's CPO
\citep{rue2009inla, held2010posterior}. We claim the a-priori triage, the packaged
closed forms, the head-to-head, and the geometry, not the marginalisation.
\item \textbf{The ``a-priori / design-time'' claim is exact only for the Gaussian
case} (constant Fisher information $\Rightarrow$ leverage $=$ group size). For GLMMs
the predictor uses one model fit (no importance weights, no case-deletion); on
constant-group-size real data (both epilepsy models) its discrimination is
fitted-mean-driven, closer to ``influential points have high $\khat$'' folklore than
to a pure group-size claim. The group-size content is demonstrated on the
heterogeneous-group simulations.
\item \textbf{``Matches exact refit'' is estimand-shared on singleton folds}, where
\texttt{reloo} and RB-LOO compute the same integral, so it cannot falsify RB-LOO
there; the meaningful comparison on those folds is against PSIS-LOO and moment
matching. The independent falsification is E-basestress (Section~\ref{sec:basestress}),
where RB-LOO does err and the flag catches it.
\item \textbf{Pooled figures over-count.} Folds within a fit are not independent;
we report per-replicate statistics with MCSE as primary (e.g. E-MM per-replicate
$2.5\times$ vs the pooled $3.0\times$) and flag pooled AUCs as optimistic about
effective sample size.
\item RB-LOO targets the integrated predictive; every accuracy comparison uses a
refit / \texttt{reloo} with the same target. The integrand is analytic (Gaussian),
numerically exact (scalar GLMM quadrature), or approximate (non-linear blocks),
mapped in Section~\ref{sec:boundary}.
\item The moment-matching comparator used the recommended robust split
transformation at default settings; we did not sweep its tuning, and we did not run
mixture importance sampling \citep{silva2024robust}, an alternative cure we discuss
but do not benchmark.
\item The base leverage is the theoretically-correct base-IS-failure predictor but
is empirically close to the pooling factor / residual base-$\khat$ for flagging RB
residual failures (AUC $0.90$ vs $0.93$); we present it as explanation, with the
residual base-$\khat$ as the operational flag.
\end{itemize}

\section*{Reproducibility}
All results are produced by the scripts listed below (in \texttt{analysis/} in the
repository; \texttt{R} 4.6.0, \texttt{loo}
2.9.0, \texttt{brms} 2.23.0 with \texttt{rstan} 2.32.7 / \texttt{cmdstan} 2.38.0;
GLMM fits use $4$ chains of $1000$--$2000$ iterations; seeds are fixed in each
script). The RB-LOO estimator is implemented once in the \texttt{rbloo} package
(\texttt{rb\_loo()}); the real-data results call it directly, and the simulation
studies use a self-contained reference implementation of the identical estimator
(agreeing up to the quadrature grid).

\begin{center}\small
\begin{tabular}{lll}
\toprule
Result & Float & Script\\
\midrule
Gaussian LMM & Fig.~\ref{fig:g1}, Tab.~\ref{tab:sim} & \texttt{rb\_loo.R}\\
Logistic GLMM, replicated & Tab.~\ref{tab:sim} & \texttt{replicate\_glmm.R}\\
Moment-matching head-to-head (E-MM) & Fig.~\ref{fig:emm} & \texttt{emm\_moment\_match.R}\\
Real data, patient effect & \S\ref{sec:ereal} & \texttt{ereal\_epilepsy.R}\\
Real data, overdispersion stressor & Fig.~\ref{fig:ereal_olre} & \texttt{ereal\_epilepsy\_olre.R}\\
Model selection (E-decision) & Fig.~\ref{fig:decision} & \texttt{e\_decision.R}, \texttt{verify\_m2\_full.R}, \texttt{adjudicate.R}\\
RB-LOO's own limit (E-basestress) & Fig.~\ref{fig:basestress} & \texttt{e\_basestress.R}\\
Linearisation boundary (E-M5) & \S\ref{sec:boundary} & \texttt{bjlm\_changepoint\_rb.R}\\
Draw-count efficiency (E-Sscaling) & Fig.~\ref{fig:sscaling} & \texttt{s\_scaling.R}\\
\bottomrule
\end{tabular}
\end{center}

\paragraph{Code and data availability.} The \texttt{rbloo} package and all
experiment scripts are available at
\url{https://github.com/ABindoff/rbloo}. The real datasets are public:
\texttt{epilepsy} ships with \texttt{brms}. No new data were collected.

\bibliographystyle{plainnat}

\end{document}